\documentclass[letterpaper, 10 pt, conference]{ieeeconf}  %
\IEEEoverridecommandlockouts                              
\usepackage{amsmath,amssymb,amsfonts}
\usepackage{graphicx}
\usepackage{bbding}
\usepackage{wrapfig}
\usepackage{pifont}
\usepackage{subfigure}
\usepackage{balance}
\usepackage[T1]{fontenc}
\usepackage{textcomp}
\usepackage{color}
\usepackage{xcolor}
\usepackage{multirow}
\usepackage{multicol,lipsum}
\usepackage{epsfig}
\usepackage{algorithm}
\usepackage{algorithmic}
\usepackage{multirow}
\usepackage{hyperref}
\usepackage{mathtools}
\usepackage{bm}
\usepackage{listings}
\newtheorem{remark}{Remark}

\newtheorem{theorem}{Theorem}

\newtheorem{lemma}{Lemma}
\newtheorem{example}{Example}
\newtheorem{corollary}{Corollary}
\newtheorem{assumption}{Assumption}
\usepackage{array}
\usepackage{float}
\usepackage[short]{optidef}
\usepackage{mdframed}
\mdfdefinestyle{mystyle}{
    backgroundcolor=white!10}

\newcommand{\oomit}[1]{}

\newcolumntype{M}[1]{>{\centering\arraybackslash}m{#1}}
\newcolumntype{N}{@{}m{0pt}@{}}

\begin{document}

\title{\LARGE \bf
Converse Barrier Certificates for Set-Based Stochastic Reach-Avoid Verification
}

\author{Bai Xue$^1$ and C.-H. Luke Ong$^2$\\
\small 1. KLSS, Institute of Software, Chinese Academy of Sciences, Beijing, China\\
2. College of Computing and Data Science, Nanyang Technological University, Singapore\\
Email: xuebai@ios.ac.cn; luke.ong@ntu.edu.sg
}

\maketitle
\thispagestyle{empty}
\pagestyle{empty}

\begin{abstract}
Recent work established sufficient and necessary barrier-like conditions for infinite-horizon reach-avoid verification of stochastic discrete-time systems from a single initial state. Whether such a converse characterization extends to a set of initial states, however, remains open. In this paper, we answer this question affirmatively for compact initial sets. We consider a uniform reach-avoid specification requiring the reach-avoid probability to exceed a prescribed threshold for every initial state in a compact set. Under appropriate assumptions, including continuous system transitions, together with a strict uniform probability margin, we extend the pointwise converse characterization to the uniform setting.
\end{abstract}

\section{Introduction}
Stochastic discrete-time systems arise widely in safety-critical applications, including autonomous vehicles, robotics, and intelligent control systems, where uncertainty and disturbances can significantly affect system behavior. For such systems, it is essential to provide formal guarantees that desired objectives are achieved with sufficiently high probability \cite{baier2008principles,franzle2011measurability}. Among these objectives, \emph{reach-avoid specifications} are particularly important: starting from an initial state, the system should reach a target region while remaining within a prescribed safe region until the target is first reached, either at some finite time without a prescribed time bound (infinite-horizon reach-avoid verification) or within a prescribed bounded time horizon (finite-horizon reach-avoid verification). Barrier certificates have emerged as an effective tool for establishing such probabilistic guarantees without explicitly computing reachable sets \cite{prajna2007framework}.

Various barrier-like conditions have been developed to facilitate the construction of barrier certificates that provide sufficient guarantees for probabilistic reach-avoid verification. However, from a practical verification perspective, the necessity direction is equally important, because a sufficient condition alone shows that the existence of a certificate implies the desired system property. When the condition is not satisfied, it is unclear whether the property itself fails or the condition is simply too restrictive. A necessary-and-sufficient characterization removes this ambiguity by establishing an equivalence between the system property and the existence of a corresponding barrier certificate under the stated assumptions. More recently, \cite{xue2026sufficient} established \emph{necessary and sufficient barrier-like conditions} for infinite-horizon reach-avoid verification from a single initial state, which we refer to as pointwise reach-avoid verification.

A natural question is whether this converse result extends from a \emph{single initial state to a set of initial states}. In many verification problems, uncertainty in the initial state, arising from sensor noise or other sources, makes the objective inherently uniform: given a compact set of possible initial states $\mathcal{X}_0$, we seek to establish
\begin{equation}
    \mathbb{P}_{\mathrm{RA}}(\bm{x}) \geq \epsilon, \qquad \forall \bm{x}\in \mathcal{X}_0.
\end{equation}
where $\mathbb{P}_{\mathrm{RA}}(\bm{x})$ denotes the probability that, starting from $\bm{x}$, the system eventually reaches the target while remaining within the safe set until the target is first reached. A key ingredient of the pointwise characterization in \cite{xue2026sufficient} is a discounted reach-avoid value function, which converges to the reach-avoid probability as the discount factor approaches one. This convergence enables the value function, for a discount factor sufficiently close to one, to serve as a barrier certificates and leads to a converse characterization for the pointwise verification problem. At first sight, one might expect the pointwise result in \cite{xue2026sufficient} to extend directly by applying it to each $\bm{x}\in \mathcal{X}_0$. However, this does not establish a uniform characterization. Indeed, for each individual $\bm{x}$, the pointwise result may provide a discount factor for which the discounted value exceeds $\epsilon$, but these state-dependent discount factors may approach one along a sequence of initial states, as noted in \cite{xue2026sufficient}. Consequently, the pointwise characterization does not by itself guarantee the existence of a common discount factor strictly smaller than one that is valid uniformly over the entire initial set.

In this paper, we resolve this issue for compact initial sets. Under appropriate assumptions, including continuity of the system dynamics in the state, measurability in the disturbance, openness of the safe and target sets, compactness of the initial set, and a strict uniform probability margin, we show that the state-dependent discount factors can be bounded uniformly away from one. The proof relies on two regularity results. First, although the reach-avoid hitting time may depend discontinuously on the initial state, we established lower semicontinuity of the discounted reach-avoid value function proposed in \cite{xue2026sufficient}. Second, we show that the minimum discount factor required to exceed the prescribed threshold is upper semicontinuous. Compactness of the initial set then yields a common discount factor strictly smaller than one. The existence of this common discount factor allows us to extend the pointwise converse characterization to the uniform setting.  Specifically, the uniform reach-avoid specification holds if and only if there exist a bounded function $v$ and a common $\gamma\in(0,1)$ satisfying the corresponding barrier-like conditions uniformly over the compact initial set. The necessity direction is obtained by constructing the certificate from the discounted reach-avoid value function associated with the common discount factor. Furthermore, based on this necessity result,  we also establish the necessity, under the same assumptions, of an alternative barrier-like formulation in \cite{xue2021reach} involving an auxiliary function $w$. Notably, this formulation eliminates the explicit discount factor from the certificate conditions, making it more convenient for computational implementation. 

The main contributions are summarized below:
\begin{enumerate}
    \item We extend the necessary-and-sufficient barrier-like characterization in \cite{xue2026sufficient} for infinite-horizon reach-avoid verification of stochastic discrete-time systems from a single initial state to a compact set of initial states.
    \item We establish the necessity of the alternative barrier-like formulation in \cite{xue2021reach}, involving an auxiliary function, for uniform infinite-horizon reach-avoid verification over a compact initial set. 
\end{enumerate}

\subsection*{Related Work}
The formal verification of stochastic systems has included qualitative almost-sure guarantees \cite{majumdar2024necessary,kordabad2026certificates} and quantitative probabilistic bounds \cite{kushner1967stochastic, prajna2007framework,chakarov2013probabilistic,wang2021safety,yu2023safe,abate2024stochastic,henzinger2025supermartingale,xue2026quantitative,chen2026construction}. For infinite-horizon verification, barrier certificates are widely used to bound safety and reach-avoid probabilities \cite{prajna2007framework,cao2025comparative}. Existing methodologies include Ville's inequality-based approaches \cite{prajna2007framework,vzikelic2023compositional} and equation-relaxation-based approaches that establish lower and upper probabilistic bounds \cite{xue2021reach,yu2023safe,xue2026sufficient}. For finite-horizon analysis, the c-martingale-based barrier framework \cite{kushner1967stochastic,steinhardt2012finite} has been developed and extended to temporal logic specifications \cite{jagtap2018temporal} and neural barrier certificates \cite{mathiesen2022safety}. More recently, \cite{zhi2024unifying,xue2024finite} have developed more general conditions for establishing lower and upper probabilistic bounds.

\noindent{\textbf{From Sufficiency to Converse Set-Based Verification.}} Most existing barrier frameworks rely on sufficient conditions and are therefore inherently conservative. Recent advances have established necessary-and-sufficient conditions for infinite-horizon safety and reach-avoid verification  of stochastic discrete-time systems via Bellman relaxations \cite{xue2026sufficient}. However, the necessity result for infinite-horizon reach-avoid verification is currently limited to the pointwise setting. Whether a barrier certificate exists for a set-based specification, namely, a uniform probabilistic guarantee over an entire continuous initial set, remains open. This work addresses this issue by establishing converse barrier certificates for set-based infinite-horizon reach-avoid verification of stochastic discrete-time systems. Under appropriate assumptions, we prove that whenever a system satisfies a reach-avoid specification uniformly over a compact initial set, a corresponding barrier certificate is guaranteed to exist. This establishes a necessary-and-sufficient characterization and removes the conservatism inherent in sufficiency-only conditions at the theoretical level.

Throughout this paper, we use the following basic notions: $\mathbb{R}$ denotes the set of real values; $\mathbb{N}$ denotes the set of nonnegative integers; for sets $\Delta_1$ and $\Delta_2$, $\overline{\Delta_1}$ denotes the closure of the set $\Delta_1$, and $\Delta_1\setminus \Delta_2$ denotes the difference of sets $\Delta_1$ and $\Delta_2$, which is the set of all elements in $\Delta_1$ that are not in $\Delta_2$;  $1_A(\bm{x})$ denotes the indicator function in the set $A$,
where, if $\bm{x}\in A$, then $1_A(\bm{x}) = 1$ and if $\bm{x}\notin A$, $1_A(\bm{x}) = 0$.

\section{Preliminaries}
\label{sec:pre}
This section introduces the stochastic discrete-time systems considered in this paper and formulated the uniform reach-avoid verification problem. We first define the system dynamics and the associated reach-avoid probability. We then recall the pointwise reach-avoid verification problem and its barrier-like characterization. Finally, we state the assumptions adopted throughout the paper and formulate the uniform reach-avoid verification problem.

\subsection{Stochastic Systems}
\label{sub:ss}
Consider the stochastic discrete-time system:
\begin{equation}
    \label{systems}
    \bm{x}_{k+1}=\bm{f}(\bm{x}_k,\bm{\theta}_k), \qquad k\in \mathbb{N},
\end{equation}
where $\bm{x}_k\in \mathbb{R}^n$ denotes the state and $\bm{\theta}_k \in \Theta \subseteq \mathbb{R}^m$ is an i.i.d. disturbance with probability distribution $\mathbb{P}_{\bm{\theta}}$. We denote by $\mathbb{E}_{\bm{\theta}}[\cdot]$ the expectation with respect to $\mathbb{P}_{\bm{\theta}}$. Here, $\Theta$ is equipped with its standard Borel $\sigma$-algebra. Let 
\[\pi=(\bm{\theta}_0,\bm{\theta}_1,\ldots)\in \Theta^{\infty}\]
denote a disturbance realization. For an initial state $\bm{x}\in \mathbb{R}^n$ and a realization $\pi\in \Theta^{\infty}$, the corresponding trajectory is denoted by $\bm{\phi}_{\bm{x}}^{\pi}: \mathbb{N}\rightarrow \mathbb{R}^n$, with $\bm{\phi}_{\bm{x}}^{\pi}(0)=\bm{x}$.

Let $\mathcal{X}\subseteq \mathbb{R}^n$ be a safe set and $\mathcal{X}_r\subseteq \mathcal{X}$ be a target set. For an initial state $\bm{x}\in \mathbb{R}^n$, define the reach-avoid hitting time 
\begin{equation}
    \label{hitting}
    \tau(\bm{x},\pi):=\inf\left\{k\in \mathbb{N} \middle|
    \begin{split}
    &\bm{\phi}_{\bm{x}}^{\pi}(k) \in \mathcal{X}_r \bigwedge \\
    &\bm{\phi}_{\bm{x}}^{\pi}(j) \in \mathcal{X}, j\leq k
    \end{split}
    \right\},
\end{equation}
where $\tau(\bm{x},\pi):=\infty$ if the target is never reached while the trajectory remains in the safe set, or if the trajectory leaves the safe set before reaching the target.

The reach-avoid probability from an initial state $\bm{x}\in \mathbb{R}^n$ is defined as 
\begin{equation}
    \label{ra_p}
\mathbb{P}_{\mathrm{RA}}(\bm{x}):=\mathbb{P}_{\pi}\bigl(\tau(\bm{x},\pi)<\infty\bigr),
\end{equation}
where $\mathbb{P}_{\pi}$ denotes the canonical product probability measure on $\bigl(\Theta^{\infty},\mathcal{B}(\Theta)^{\otimes\infty}\bigr)$ induced by the i.i.d. disturbance distribution $\mathbb{P}_{\bm{\theta}}$ via the Kolmogorov extension theorem. We denote by $\mathbb{E}_{\pi}[\cdot]$ the expectation with respect to $\mathbb{P}_{\pi}$. By definition, $\mathbb{P}_{\mathrm{RA}}(\bm{x})=0$ for all $\bm{x}\in\mathbb{R}^n\setminus\mathcal{X}$ and $\mathbb{P}_{\mathrm{RA}}(\bm{x})=1$ for all $\bm{x}\in \mathcal{X}_r$.

\subsection{Pointwise Reach-Avoid Verification}
We recall the pointwise reach-avoid verification and its sufficient and necessary barrier-like condition in \cite{xue2026sufficient}.

Given $\bm{x}_0\in \mathcal{X}\setminus\mathcal{X}_r$ and a prescribed probability threshold $\epsilon\in(0,1)$, pointwise reach-avoid verification aims to certify that, starting from $\bm{x}_0$, the probability of reaching $\mathcal{X}_r$ while remaining in $\mathcal{X}$ until $\mathcal{X}_r$ is first reached is at least $\epsilon$, i.e., 
\[\mathbb{P}_{\mathrm{RA}}(\bm{x}_0)\geq \epsilon.\]

A discounted reach-avoid value function was used in \cite{xue2026sufficient} for the pointwise  reach-avoid verification. For $\bm{x}\in \mathbb{R}^n$ and $\gamma \in (0,1)$, define 
\begin{equation}
\label{dis_val_fun}
\widetilde{V}_{\gamma}(\bm{x}):=\mathbb{E}_{\pi}[\gamma^{\tau(\bm{x},\pi)}],
\end{equation}
with the convention $\gamma^{\infty}=0$. Moreover, 
\[\lim_{\gamma \rightarrow 1^-}\widetilde{V}_{\gamma}(\bm{x})=\mathbb{P}_{\mathrm{RA}}(\bm{x}).\]
The value function $\widetilde{V}_{\gamma}:\mathbb{R}^n\rightarrow \mathbb{R}$ is the unique bounded solution to the Bellman equation
\begin{equation}
\widetilde{V}_{\gamma}(\bm{x})=
\begin{cases}
1, & \bm{x}\in\mathcal{X}_r,\\[1mm]
\gamma\mathbb{E}_{\theta}\!\left[
\widetilde{V}_{\gamma}\bigl(\bm{f}(\bm{x},\bm{\theta})\bigr)
\right],
& \bm{x}\in\mathcal{X}\setminus\mathcal{X}_r,\\[1mm]
0, & \bm{x}\notin\mathcal{X}.
\end{cases}
\end{equation}
Based on this equation, \cite{xue2026sufficient} established the following necessary and sufficient barrier-like characterization for the pointwise reach-avoid verification problem.

\begin{theorem}
    Let $\epsilon\in [0,1)$. Given $\bm{x}_0\in \mathcal{X}\setminus \mathcal{X}_r$, if there exist a constant $\gamma \in (0,1)$ and a barrier certificate $v:\mathbb{R}^n \rightarrow \mathbb{R}$, which is bounded over $\mathcal{X}$ and satisfies the following barrier-like condition: 
\begin{equation}
\label{constraint_gamma0}
    \begin{cases}
        v(\bm{x}_0)\geq \epsilon,\\
         v(\bm{x}) \leq \gamma \mathbb{E}_{\bm{\theta}}[v(\bm{f}(\bm{x},\bm{\theta}))], & \forall \bm{x}\in \mathcal{X}\setminus \mathcal{X}_r,\\
         v(\bm{x})\leq 1, &\forall \bm{x}\in \mathcal{X}_r, \\
         v(\bm{x})\leq 0, & \forall \bm{x}\in \mathbb{R}^n\setminus \mathcal{X},
    \end{cases}
\end{equation}
then, $\mathbb{P}_{\mathrm{RA}}(\bm{x}_0)\geq \epsilon$.  Moreover, if $\mathbb{P}_{\mathrm{RA}}(\bm{x}_0)>\epsilon$, there indeed exist such a constant 
$\gamma \in (0,1)$ and a barrier certificate $v:\mathbb{R}^n \rightarrow \mathbb{R}$, bounded over $\mathcal{X}$, that satisfy \eqref{constraint_gamma0}.
\end{theorem}

According to \cite{xue2026sufficient}, for any $\gamma\in(0,1)$, if
$\widetilde{V}_{\gamma}:\mathbb{R}^n\to\mathbb{R}$ satisfies $\widetilde{V}_{\gamma}(\bm{x}_0)\geq\epsilon$, then the pair $(\widetilde{V}_{\gamma},\gamma)$ satisfies the barrier-like condition \eqref{constraint_gamma0}. Since
\[
\mathbb{P}_{\mathrm{RA}}(\bm{x}_0)>\epsilon
\quad\text{and}\quad
\lim_{\gamma\rightarrow1^-}
\widetilde{V}_{\gamma}(\bm{x}_0)
=
\mathbb{P}_{\mathrm{RA}}(\bm{x}_0),
\]
there exists $\gamma\in(0,1)$ sufficiently close to one such that
$\widetilde{V}_{\gamma}(\bm{x}_0)>\epsilon$. Hence, the pair $(\widetilde{V}_{\gamma},\gamma)$ satisfying
\eqref{constraint_gamma0} exists.

Although this barrier-like condition \eqref{constraint_gamma0}, with $v(\bm{x})\geq \epsilon, \forall \bm{x}\in \mathcal{X}_0$ replacing $v(\bm{x}_0)\geq \epsilon$, can also be applied to certify the reach-avoid probability over a set of initial states $\mathcal{X}_0\subseteq \mathcal{X}\setminus \mathcal{X}_r$, its necessity is still not guaranteed. The main obstacle is that the discount factor $\gamma$ in \eqref{constraint_gamma0} may depend on the initial state, making it difficult to obtain a single $\gamma<1$ that is valid uniformly over the entire initial set $\mathcal{X}_0$.

\subsection{Uniform Reach-Avoid Verification}
\label{sec:uniform}

We now consider the reach-avoid verification problem over a set of initial
states $\mathcal{X}_0\subseteq\mathcal{X}\setminus\mathcal{X}_r$. 
Given a prescribed probability threshold
$\epsilon\in(0,1)$, the uniform reach-avoid verification problem is to
establish
\begin{equation}
\label{eq:uniform}
\mathbb{P}_{\mathrm{RA}}(\bm{x})\geq\epsilon,
\qquad
\forall \bm{x}\in\mathcal{X}_0.
\end{equation}
Equivalently,
\begin{equation}
\label{eq:uniform-inf}
\inf_{\bm{x}\in\mathcal{X}_0}
\mathbb{P}_{\mathrm{RA}}(\bm{x})
\geq\epsilon.
\end{equation}

The objective of this paper is to establish the existence of a barrier certificate $v:\mathbb{R}^n \rightarrow \mathbb{R}$ and $\gamma\in (0,1)$, satisfying \eqref{constraint_gamma0} with $v(\bm{x})\geq \epsilon, \forall \bm{x}\in \mathcal{X}_0$ replacing $v(\bm{x}_0)\geq \epsilon$, for the uniform reach-avoid specification
\eqref{eq:uniform} under the assumptions stated below.

\begin{assumption}
\label{ass}
\begin{enumerate}
\item For every $\bm{\theta} \in \Theta$, $\bm{f}(\cdot,\bm{\theta}):\mathbb{R}^n \rightarrow\mathbb{R}^n$
is continuous, and for every $\bm{x}\in \mathbb{R}^n$, $\bm{f}(\bm{x},\cdot): \Theta \rightarrow \mathbb{R}^n$ is measurable.

\item Both the safe set $\mathcal{X}\subset\mathbb{R}^n$ and the target set $\mathcal{X}_r\subset\mathcal{X}$ are open.

\item The initial set
\[
\mathcal{X}_0\subseteq\mathcal{X}\setminus\mathcal{X}_r
\]
is compact.

\item There is a strict reach-avoid margin:
\begin{equation}
\label{eq:strict-margin}
\mathbb{P}_{\mathrm{RA}}(\bm{x})>\epsilon,
\qquad
\forall\bm{x}\in\mathcal{X}_0.
\end{equation}
\end{enumerate}
\end{assumption}

The openness of the safe set reflects the convention that reaching its boundary constitutes leaving the safe region. The compactness of $\mathcal{X}_0$, together with the continuity of $\bm{f}$ with respect to $\bm{x}$ for every fixed $\bm{\theta} \in \Theta$, is essential for establishing a uniform characterization over the initial set.

In addition, under Assumption \ref{ass}, for every fixed $\bm{x}\in\mathbb{R}^n$, the finite-horizon trajectory map $\pi\mapsto\bm{\phi}_{\bm{x}}^{\pi}(k)$ is measurable for every $k\in\mathbb{N}$. Hence, for each $k\in \mathbb{N}$,
\[
\begin{split}
&\{\pi\mid \tau(\bm{x},\pi)=k\}\\
=&\bigcap_{j=0}^{k-1}
\{\pi \mid \bm{\phi}_{\bm{x}}^{\pi}(j)\in X\setminus \mathcal{X}_r\}
\cap
\{\pi \mid \bm{\phi}_{\bm{x}}^{\pi}(k)\in X_r\}
\end{split}
\]
is measurable, since $\mathcal{X}$ and $\mathcal{X}_r$ are open and hence Borel sets. Thus $\tau(\bm{x},\cdot)$ is a measurable random variable, and consequently $\mathbb{P}_{\mathrm{RA}}(\bm{x})$ and $\widetilde{V}_\gamma(\bm{x})$ are well-defined.

\section{Necessity of Uniform Reach-Avoid Verification}
\label{sec:urav}
In this section, we establish the necessity of the barrier-like conditions for uniform reach-avoid verification. In order to address this issue, we first establish the lower semicontinuity of the discounted value function \eqref{dis_val_fun} and then show that, under the strict reach-avoid margin, a single discount factor $\gamma \in (0,1)$ can be chosen uniformly over the compact initial set $\mathcal{X}_0$. This uniform discount factor allows the discounted reach-avoid value function itself to be constructed as a uniform barrier certificate.

\subsection{Lower Semicontinuity of the Discounted Value Function}
\label{sec:ldv}

We first establish the lower semicontinuity of the discounted reach-avoid value function \eqref{dis_val_fun} with respect to the initial state. The key observation is that,for a fixed disturbance realization, the discounted hitting-time payoff $\gamma^{\tau(\bm{x},\pi)}$ is lower semicontinuous with respect to the initial state. Fatou's lemma then allows this pathwise property to be transferred to the expected discounted value.

\begin{lemma}
\label{lemma1}
Under Assumption~\ref{ass}, for every $\gamma\in(0,1)$ and every $\bm{x}_0\in\mathcal{X}_0$, the function
\[\bm{x}\mapsto\gamma^{\tau(\bm{x},\pi)}\] is lower semicontinuous at $\bm{x}_0$ for  every disturbance realization $\pi$.
\end{lemma}
\begin{proof}
    Fix $x_0\in\mathcal{X}_0$ and a disturbance realization $\pi$.

If $\tau(\bm{x}_0,\pi)=\infty$, then, by the convention $\gamma^{\infty}=0$,
\[\gamma^{\tau(\bm{x}_0,\pi)}=0.\] 
Since $\gamma^{\tau(\bm{x},\pi)}\geq 0$ for every $\bm{x}$, we have 
\[\liminf_{\bm{x}\rightarrow \bm{x}_0} \gamma^{\tau(\bm{x},\pi)}\geq 0=\gamma^{\tau(\bm{x}_0,\pi)},\]
and hence lower semicontinuity holds at $\bm{x}_0$.

Now suppose that $\tau(\bm{x}_0,\pi)=k<\infty$. By definition, 
\[\bm{\phi}_{\pi}^{\bm{x}_0}(j)\in\mathcal{X}, \qquad j=0,\ldots,k-1,\]
and 
\[\bm{\phi}_{\pi}^{\bm{x}_0}(k)\in\mathcal{X}_r.\] 

Since $\mathcal{X}$ and $\mathcal{X}_r$ are open and $\bm{f}$ is continuous, the finite-horizon trajectory $\bm{\phi}_{\pi}^{\bm{x}}(j)$ depends continuously on $\bm{x}$ for every $j\leq k$. Hence there exists a neighborhood $\mathcal{U}$ of $\bm{x}_0$ such that, for every $\bm{x}\in \mathcal{U}$,
\[\bm{\phi}_{\pi}^{\bm{x}}(j)\in\mathcal{X}, \qquad j=0,\ldots,k-1,\]
and 
\[\bm{\phi}_{\pi}^{\bm{x}}(k)\in\mathcal{X}_r.\]
Consequently,
\[\tau(\bm{x},\pi)\leq k\]
for every $\bm{x}\in \mathcal{U}$. Since $\gamma\in(0,1)$,
\[\gamma^{\tau(\bm{x},\pi)} \geq \gamma^k = \gamma^{\tau(\bm{x}_0,\pi)}.\]
Thus, 
\[\liminf_{\bm{x}\rightarrow \bm{x}_0}\gamma^{\tau(\bm{x},\pi)} \geq \gamma^{\tau(\bm{x}_0,\pi)}.\]

Thus $\bm{x}\mapsto\gamma^{\tau(\bm{x},\pi)}$ is lower semicontinuous at $\bm{x}_0$.
\end{proof}

\begin{lemma}
    \label{lemma2}
Under Assumption \ref{ass}, the discounted reach-avoid value function $\widetilde{V}_{\gamma}$ is lower semicontinuous on $\mathcal{X}_0$ for every $\gamma\in(0,1)$.
\end{lemma}
\begin{proof}
    Fix $\bm{x}_0\in\mathcal{X}_0$. By Lemma \ref{lemma1},
    \[\bm{x}\mapsto\gamma^{\tau(\bm{x},\pi)}\]
    is lower semicontinuous at $\bm{x}_0$ for every $\pi$. Since
    \[0\leq\gamma^{\tau(\bm{x},\pi)}\leq1,\]
    Fatou's lemma gives
    \[
    \begin{aligned}
    \liminf_{\bm{x}\rightarrow \bm{x}_0}\widetilde{V}_{\gamma}(\bm{x})&\geq \mathbb{E}_{\pi}\left[\liminf_{\bm{x}\rightarrow \bm{x}_0} \gamma^{\tau(\bm{x},\pi)}\right]\\
&\geq \mathbb{E}_{\pi}\left[ \gamma^{\tau(\bm{x}_0,\pi)}\right]=\widetilde{V}_{\gamma}(\bm{x}_0).
\end{aligned}
\]
Hence $\widetilde{V}_{\gamma}$ is lower semicontinuous at $\bm{x}_0$. Since $\bm{x}_0$ is arbitrary, the result follows.
\end{proof}

\subsection{Existence of a Uniform Discount Factor}
\label{sec:eudf}
For each initial state $\bm{x}\in \mathcal{X}_0$,
the discounted reach-avoid value defined in \eqref{dis_val_fun} converges to the reach-avoid probability
as $\gamma\uparrow1$, i.e., $\lim_{\gamma \rightarrow 1}\widetilde{V}(\bm{x})=\mathbb{P}_{\mathrm{RA}}(\bm{x}), \forall \bm{x}\in \mathcal{X}_0$. The strict margin therefore guarantees that each
initial state admits a discount factor for which the discounted value
exceeds $\epsilon$ \cite{xue2026sufficient}. We in this section show that these state-dependent discount factors
can be bounded uniformly away from one by exploiting the lower
semicontinuity of the discounted value function \eqref{dis_val_fun}  and the compactness of
$\mathcal{X}_0$.

For each $\bm{x}\in\mathcal{X}_0$, define 
\begin{equation}
    \label{eq:gamma-threshold} 
    \gamma_{\epsilon}(\bm{x}) := \inf \left\{ \gamma\in(0,1)\mid \widetilde{V}_{\gamma}(\bm{x})>\epsilon \right\}. 
\end{equation}
By the strict margin in Assumption \ref{ass} and the convergence 
\[
\lim_{\gamma \uparrow 1} \widetilde{V}_{\gamma}(\bm{x}) = \mathbb{P}_{\mathrm{RA}}(\bm{x}) > \epsilon, \qquad \forall\bm{x}\in\mathcal{X}_0, 
\]
the set in \eqref{eq:gamma-threshold} is nonempty. Hence, 
\[
\gamma_{\epsilon}(\bm{x})<1, \qquad \forall\bm{x}\in\mathcal{X}_0.
\]
\begin{lemma}
     \label{lemma3} 
 Under Assumption \ref{ass}, the function \[ \gamma_{\epsilon}:\mathcal{X}_0\rightarrow[0,1) 
     \] is upper semicontinuous. 
\end{lemma}
\begin{proof}
    Fix $c\in(0,1)$. We first show that 
    \begin{equation} 
\label{eq:gamma-equivalence} 
\gamma_{\epsilon}(\bm{x})<c \quad\Longleftrightarrow\quad \widetilde{V}_{c}(\bm{x})>\epsilon. 
\end{equation} 

Suppose that $\gamma_{\epsilon}(\bm{x})<c$. By the definition of $\gamma_{\epsilon}$, there exists some $\gamma<c$ such that 
\[ 
\widetilde{V}_{\gamma}(\bm{x})>\epsilon. 
\]
Since $\widetilde{V}_{\gamma}(\bm{x})$ is nondecreasing in $\gamma$, 
\[
\widetilde{V}_{c}(\bm{x}) \geq \widetilde{V}_{\gamma}(\bm{x}) > \epsilon. 
\]

Conversely, suppose that $\widetilde{V}_{c}(\bm{x})>\epsilon$. For fixed $\bm{x}$, the map 
\[\gamma \mapsto \widetilde{V}_{\gamma}(\bm{x}) = \mathbb{E}_{\pi} \left[ \gamma^{\tau(\bm{x},\pi)} \right] \]
is continuous on $(0,1)$ \cite{xue2026sufficient}. Hence, for some $\gamma<c$ sufficiently close to $c$,
\[ 
\widetilde{V}_{\gamma}(\bm{x})>\epsilon,
\]
which implies $\gamma_{\epsilon}(\bm{x})<c$. Thus, \eqref{eq:gamma-equivalence} holds.

By Lemma \ref{lemma2}, the set 
\[\left\{ \bm{x}\in\mathcal{X}_0\mid \widetilde{V}_{c}(\bm{x})>\epsilon \right\} \]
is relatively open in $\mathcal{X}_0$. By \eqref{eq:gamma-equivalence}, this set coincides with 
\[
\left\{ \bm{x}\in\mathcal{X}_0\mid \gamma_{\epsilon}(\bm{x})<c \right\}.
\]
Thus, for every $c\in(0,1)$, the strict sublevel set
$\{\bm{x}\in\mathcal{X}_0:\gamma_{\epsilon}(\bm{x})<c\}$
is relatively open in $\mathcal{X}_0$

Therefore, $\gamma_{\epsilon}$ is upper semicontinuous. 
\end{proof}

\begin{lemma}[Uniform Discount Factor] 
\label{lemma4}
 Under Assumption~\ref{ass}, there exists $\bar{\gamma}\in(0,1)$ such that 
\begin{equation} 
\label{eq:uniform-discount} 
\widetilde{V}_{\bar{\gamma}}(\bm{x})>\epsilon, \qquad \forall\bm{x}\in\mathcal{X}_0. 
\end{equation} 
\end{lemma} 
\begin{proof}
    Since $\gamma_{\epsilon}$ is upper semicontinuous on the compact set $\mathcal{X}_0$, it attains its maximum. Let 
    \[
    \gamma^* := \max_{\bm{x}\in\mathcal{X}_0} \gamma_{\epsilon}(\bm{x}).
    \]

    Because $\gamma_{\epsilon}(\bm{x})<1$ for every $\bm{x}\in\mathcal{X}_0$, we have 
\[
\gamma^*<1. 
\]
 Choose any  $\bar{\gamma}\in(\gamma^*,1)$.
 Then,
\[
\bar{\gamma}>\gamma_{\epsilon}(\bm{x}), \qquad \forall\bm{x}\in\mathcal{X}_0.
\]
By the definition of $\gamma_{\epsilon}$ and the monotonicity of $\widetilde{V}_{\gamma}(\bm{x})$ with respect to $\gamma$, 
\[
\widetilde{V}_{\bar{\gamma}}(\bm{x})>\epsilon, \qquad \forall\bm{x}\in\mathcal{X}_0. 
\]
We complete the proof.
\end{proof}

\begin{remark} 
The strict margin in Assumption~\ref{ass} is essential for obtaining a common discount factor $\bar{\gamma}<1$ over the compact initial set $\mathcal{X}_0$. Without this margin, the discount factor required to satisfy $\widetilde{V}_{\gamma}(\bm{x})>\epsilon$ may approach one along a sequence of initial states. In that case, a single discount factor strictly smaller than one may not exist, and the pointwise necessity result does not directly extend to the uniform setting. 
\end{remark}

\subsection{Necessity of Uniform Barrier-Like Conditions} \label{sec:uniform-necessity}
We now establish the necessity of the barrier-like conditions for the uniform reach-avoid specification. The key idea is to use the common discount factor obtained in Lemma~\ref{lemma4} and construct the certificate directly from the corresponding discounted reach-avoid value function \eqref{dis_val_fun}. Its Bellman equation yields the one-step barrier-like condition, while its prescribed values on the target and unsafe regions yield the remaining conditions. Thus, the only additional ingredient required to lift the pointwise converse to the uniform setting is the existence of a common discount factor over the compact initial set.

\begin{theorem}[Necessity of Uniform Barrier-Like Conditions] 
\label{thm:uniform-necessity} 
Under Assumption~\ref{ass}, there exist a constant $\gamma \in (0,1)$ and a barrier certificate $v:\mathbb{R}^n \rightarrow \mathbb{R}$, which is bounded over $\mathcal{X}$ and satisfies the following barrier-like condition:  
\begin{equation} 
\label{eq:uniform-barrier-conditions}
 \begin{cases} 
v(\bm{x})\geq\epsilon, & \bm{x}\in\mathcal{X}_0,\\
v(\bm{x})\leq \gamma\mathbb{E}_{\bm{\theta}} \left[ v\bigl(\bm{f}(\bm{x},\bm{\theta})\bigr) \right], & \bm{x}\in\mathcal{X}\setminus\mathcal{X}_r,\\
v(\bm{x})\leq1, & \bm{x}\in\mathcal{X}_r,\\
v(\bm{x})\leq0, & \bm{x}\in \mathbb{R}^n\setminus\mathcal{X}. \
\end{cases} 
\end{equation} 
\end{theorem} 
\begin{proof}
    By Lemma \ref{lemma4}, there exists $\bar{\gamma}\in(0,1)$ such that 
\[ 
\widetilde{V}_{\bar{\gamma}}(\bm{x})>\epsilon, \qquad \forall\bm{x}\in\mathcal{X}_0. 
\]

Define 
\[ 
v(\bm{x}) := \widetilde{V}_{\bar{\gamma}}(\bm{x}), \qquad \bm{x}\in \mathbb{R}^n.
\] For every $\bm{x}\in \mathbb{R}^n$, the Bellman equation for the discounted reach-avoid value function gives 
\[ 
v(\bm{x}) = 1_{\mathcal{X}_r}(\bm{x})+1_{\mathcal{X}\setminus\mathcal{X}_r}(\bm{x})\bar{\gamma} \mathbb{E}_{\bm{\theta}} \left[ v\bigl(\bm{f}(\bm{x},\bm{\theta})\bigr) \right]. 
\] 

Then, following the proof of Theorem 3 in \cite{xue2026sufficient}, the function  $v$ satisfies \eqref{eq:uniform-barrier-conditions}, 
establishing the conclusion.
\end{proof}

    Based on the value function $\widetilde{V}_{\bar{\gamma}}$, we are able to show the necessity of another sufficient barrier-like condition in \cite{xue2021reach} for the reach-avoid verification under Assumption \ref{ass}. The condition is presented below:
\begin{equation}
\label{w}
    \begin{cases}
        v(\bm{x})\geq \epsilon,& \forall \bm{x}\in \mathcal{X}_0,\\
        v(\bm{x})\leq  \mathbb{E}_{\bm{\theta}}[v(\bm{f}(\bm{x},\bm{\theta}))],& \forall \bm{x}\in \mathcal{X}\setminus \mathcal{X}_r,\\
        v(\bm{x})\leq \mathbb{E}_{\bm{\theta}}[w(\bm{f}(\bm{x},\bm{\theta}))]-w(\bm{x}), & \forall \bm{x}\in \mathcal{X}\setminus \mathcal{X}_r,\\
        v(\bm{x})\leq 1, & \forall \bm{x}\in \mathcal{X}_r,\\
        v(\bm{x})\leq 0, & \forall \bm{x}\in \mathbb{R}^n\setminus \mathcal{X}.
    \end{cases}
\end{equation}
 If there exist a function $v:\mathbb{R}^n\rightarrow \mathbb{R}$ and a bounded function $w:\mathbb{R}^n\rightarrow \mathbb{R}$ satisfying \eqref{w}, $\mathbb{P}_{\mathrm{RA}}({\bm{x}_0})\geq \epsilon$ holds for all $\bm{x}_0\in \mathcal{X}_0$. This conclusion can be justified by following the proof of Corollary 2 in \cite{xue2021reach}. In the following, we just demonstrate its necessity. 
\begin{corollary}
\label{cor:w}
Under Assumption \ref{ass}, there exist a function $v:\mathbb{R}^n\rightarrow \mathbb{R}$ and a bounded function $w:\mathbb{R}^n\rightarrow \mathbb{R}$ satisfying barrier-like condition \eqref{w}.
\end{corollary}
\begin{proof}
The proof follows the one of Corollary 1 in \cite{xue2026sufficient}.
We can obtain 
   \[
   \begin{cases}
   \tilde{V}_{\bar{\gamma}}(\bm{x})\geq \epsilon,& \forall \bm{x}\in \mathcal{X}_0,\\
    1\geq \tilde{V}_{\bar{\gamma}}(\bm{x})\geq 0,& \forall \bm{x}\in \mathbb{R}^n,\\
   \tilde{V}_{\bar{\gamma}}(\bm{x})=\bar{\gamma} \mathbb{E}_{\bm{\theta}}[\tilde{V}_{\bar{\gamma}}(\bm{f}(\bm{x},\bm{\theta}))]\\
   ~~~~~~~~~~~\leq \mathbb{E}_{\bm{\theta}}[\tilde{V}_{\bar{\gamma}}(\bm{f}(\bm{x},\bm{\theta}))], & \forall \bm{x}\in \mathcal{X}\setminus \mathcal{X}_r,\\
   \tilde{V}_{\bar{\gamma}}(\bm{x})\leq 1, & \forall \bm{x}\in \mathcal{X}_r,\\
   \tilde{V}_{\bar{\gamma}}(\bm{x})=0,  & \forall \bm{x}\in \mathbb{R}^n\setminus \mathcal{X}.
   \end{cases}
   \]
   Let $\gamma_1$ be a constant satisfying $\frac{\gamma_1}{1+\gamma_1}\geq \bar{\gamma}$, and $w(\bm{x}):=\gamma_1 \tilde{V}_{\bar{\gamma}}(\bm{x})$ for $\bm{x}\in \mathbb{R}^n$. Thus,  
   \begin{equation*}
   \begin{split}
   &\frac{\mathbb{E}_{\bm{\theta}}[w(\bm{f}(\bm{x},\bm{\theta}))]-w(\bm{x})-\tilde{V}_{\bar{\gamma}}(\bm{x})}{1+\gamma_1}\\
   =&\frac{\gamma_1 \mathbb{E}_{\bm{\theta}}[\tilde{V}_{\bar{\gamma}}(\bm{f}(\bm{x},\bm{\theta}))]-\gamma_1\tilde{V}_{\bar{\gamma}}(\bm{x})-\tilde{V}_{\bar{\gamma}}(\bm{x})}{1+\gamma_1}\\
   =&\frac{\gamma_1}{1+\gamma_1} \mathbb{E}_{\bm{\theta}}[\tilde{V}_{\bar{\gamma}}(\bm{f}(\bm{x},\bm{\theta}))]-\tilde{V}_{\bar{\gamma}}(\bm{x})\\
   \geq &\bar{\gamma} \mathbb{E}_{\bm{\theta}}[\tilde{V}_{\bar{\gamma}}(\bm{f}(\bm{x},\bm{\theta}))]-\tilde{V}_{\bar{\gamma}}(\bm{x})=0, \qquad \quad \forall \bm{x}\in \mathcal{X}\setminus \mathcal{X}_r.
   \end{split}
   \end{equation*}
   Thus, the functions $\tilde{V}_{\bar{\gamma}}(\bm{x})$ and $w(\bm{x}):=\gamma_1 \tilde{V}_{\bar{\gamma}}(\bm{x})$ satisfy \eqref{w}. Consequently,  there exist a function $v:\mathbb{R}^n\rightarrow \mathbb{R}$ and a bounded function $w:\mathbb{R}^n\rightarrow \mathbb{R}$ satisfying \eqref{w}.  
   \end{proof}

The barrier-like conditions in Theorem \ref{thm:uniform-necessity} and Corollary \ref{cor:w} differ mainly in how the discount factor is incorporated. Theorem \ref{thm:uniform-necessity}  uses a single certificate $v$ together with an explicit discount factor $\gamma\in(0,1)$ through the condition
\[
v(\bm{x})\le \gamma\mathbb{E}_{\bm{\theta}}\!\left[v(\bm{f}(\bm{x},\bm{\theta}))\right].
\]
This condition is bilinear in $(\gamma,v)$ and is therefore not jointly convex when both $\gamma$ and $v$ are treated as decision variables. However, for a fixed $\gamma$, the condition is affine in $v$. In contrast, Corollary \ref{cor:w} eliminates the explicit discount factor by introducing an auxiliary function $w$ and imposing
\[
\begin{cases}
v(\bm{x})\le \mathbb{E}_{\bm{\theta}}\!\left[v(\bm{f}(\bm{x},\bm{\theta}))\right], & \forall \bm{x}\in \mathcal{X}\setminus \mathcal{X}_r,\\
v(\bm{x})\le \mathbb{E}_{\bm{\theta}}\!\left[w(\bm{f}(\bm{x},\bm{\theta}))\right]-w(\bm{x}), & \forall \bm{x}\in \mathcal{X}\setminus \mathcal{X}_r.
\end{cases}
\]
Thus, Corollary \ref{cor:w} replaces the multiplicative discount factor with an auxiliary function and yields constraints that are affine in the decision functions $v$ and $w$, avoiding the bilinearity associated with jointly optimizing $\gamma$ and $v$. A related comparison between these two formulations is also discussed in \cite{cao2025comparative}.

\begin{remark}
In the above barrier-like conditions, $\mathbb{R}^n$ can be replaced by a subset $\widehat{\mathcal{X}}$ satisfying
\begin{equation}
\label{eq:xhat}
\overline{\mathcal{X}}
\cup
\bm{f}(\overline{\mathcal{X}},\Theta)
\subseteq
\widehat{\mathcal{X}}.
\end{equation}
with $\bm{f}(\overline{\mathcal{X}},\Theta):=
\left\{\bm{f}(\bm{x},\bm{\theta})\mid
\bm{x}\in\overline{\mathcal{X}},
\bm{\theta}\in\Theta
\right\}$.
This allows the barrier-like conditions to be enforced only on the computational domain $\widehat{\mathcal{X}}$ rather than on the entire state space $\mathbb{R}^n$. As shown in \cite{xue2021reach,xue2024finite}, this restriction is justified by considering a switched system that keeps the state fixed once it leaves the safe set. Consequently, it is sufficient for $\widehat{\mathcal{X}}$ to contain $\overline{\mathcal{X}}$ and its one-step successor set $\bm{f}(\overline{\mathcal{X}},\Theta)$, since states outside the safe set no longer evolve. 
\end{remark}

\section{Examples}
\label{sec:ex}
In this section, we demonstrate the application of our theoretical developments through one example. In this case, the function $\bm{f}(\bm{x},\bm{\theta})$ is a polynomial in the state variables $x$, and the safe set $\mathcal{X}$ and the target set $\mathcal{X}_r$ are semi-algebraic sets. We aim to search for polynomial barrier certificates that satisfy the constraints in \eqref{eq:uniform-barrier-conditions} and \eqref{w}. To do this, we encode the constraints \eqref{eq:uniform-barrier-conditions} and \eqref{w} as semi-definite programs (SDPs) using the sum-of-squares (SOS) decomposition for multivariate polynomials \cite{parrilo2003semidefinite}. The resulting SDPs are then solved  using the tool Mosek 10.1.21 \cite{aps2019mosek}. To ensure numerical stability during the solution of these SDPs, we impose a constraint on the coefficients of the unknown polynomials, specifically restricting them to the interval $[-100, 100]$. In the sequel, $\sum[\bm{x}]$ denotes the set of SOS polynomials over variables $\bm{x}$, i.e., 
$\sum[\bm{x}]=\{p\in \mathbb{R}[\bm{x}]\mid p=\sum_{i=1}^k q^2_i(\bm{x}), q_i(\bm{x})\in \mathbb{R}[\bm{x}],i=1,\ldots,k\}$, where $\mathbb{R}[\bm{x}]$ denotes the ring of polynomials in variables $\bm{x}$. For more detailed numerical demonstrations and comparisons, please refer to \cite{cao2025comparative}.

In addition, to empirically validate the reach-avoid guarantees, we employ a parallelized Monte Carlo approach. We systematically sample $10^4$ initial states in $\mathcal{X}_0$  using a polar grid and simulate $10^4$ independent trajectories from each initial state for $10^3$ steps under uniformly distributed disturbances. Each trajectory is monitored to verify that it remains within the safe set $\mathcal{X}$ until reaching the target set $\mathcal{X}_r$. We then compute the empirical reach-avoid probability for each initial state and estimate the worst-case probability by taking the infimum over the sampled initial states.  

\begin{example}
\label{ex1}
We consider the following system, which is from \cite{xue2021reach},
\begin{equation*}
\begin{cases}
x(l+1)=x(l)+0.01(-0.5x(l)-0.5y(l)+0.5x(l)y(l)),\\
y(l+1)=y(l)+0.01(-0.5y(l)+1+\theta(l)),
\end{cases}
\end{equation*}
where $\theta(l)\in\Theta=[-10,10]$ is uniform,  $\mathcal{X}_0=\{\bm{x}\mid h_0(x,y) \leq 0\}$ with $h_0(\bm{x})=x^2+(y+0.7)^2-0.01$, $\mathcal{X}=\{\bm{x}\mid h(\bm{x})< 0\}$ with $h(\bm{x})=x^2+y^2-1$,  $\mathcal{X}_r=\{\bm{x}\mid g(\bm{x})<0\}$ with $g(\bm{x})=10x^2+10(y-0.5)^2-1$, and $\widehat{\mathcal{X}}=\{\bm{x}\mid \widehat{h}(\bm{x})\leq 0\}$ with $\widehat{h}(\bm{x})=x^2+y^2-1.1$.  The Monte Carlo (\textbf{MC}) simulations yield an empirical worst-case reach-avoid probability of $0.8445$ over the sampled initial states. Two system trajectories, originating in the initial set $\mathcal{X}_0$, are visualized on Fig. \ref{fig:ex1_2}.

\begin{figure}[h]
\centering
\includegraphics[width=0.4\textwidth, height=5cm]{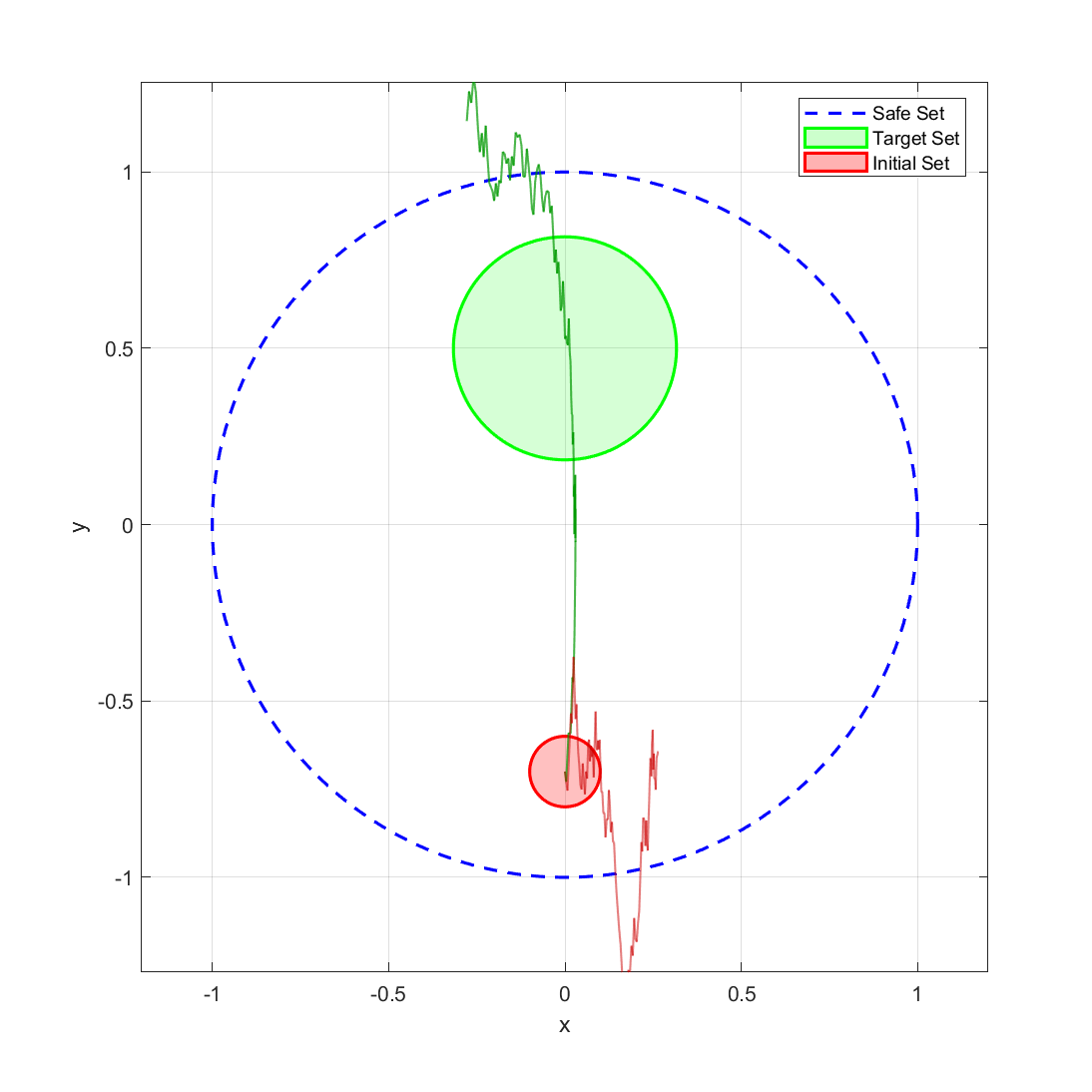}
\caption{An illustration of two system trajectories.}
\label{fig:ex1_2}
\end{figure}

We first consider the discounted barrier-like condition \eqref{eq:uniform-barrier-conditions}. For a prescribed discount factor $\gamma$ and $\epsilon$, we search for a polynomial function $v$ and SOS multipliers satisfying
     \begin{equation}
     \label{v_sos}
\begin{cases}
   v(\bm{x})-\epsilon+s_0(\bm{x})h_0(\bm{x})\in \sum[\bm{x}], \\
       \gamma \mathbb{E}_{\theta}[v(\bm{f}(\bm{x},\bm{\theta}))]-v(\bm{x})+s_1(\bm{x}) h(\bm{x})\\
       \qquad\qquad\qquad\qquad\quad-s_2(\bm{x}) g(\bm{x}) \in \sum[\bm{x}],\\
        1-v(\bm{x})+s_3(\bm{x})g(\bm{x})\in \sum[\bm{x}],\\
       -v(\bm{x})-s_4(\bm{x})h(\bm{x})+s_5(\bm{x})\widehat{h}(\bm{x})\in \sum[\bm{x}],\\
       s_0(\bm{x})\in \sum[\bm{x}],s_1(\bm{x})\in \sum[\bm{x}],\\
       s_2(\bm{x})\in \sum[\bm{x}],s_3(\bm{x})\in \sum[\bm{x}],\\
       s_4(\bm{x})\in \sum[\bm{x}], s_5(\bm{x})\in \sum[\bm{x}],
       \end{cases}
\end{equation}

The feasibility results are summarized in Table \ref{tab:sdp_ex1_set}. The table shows that feasibility strongly depends on both the polynomial degree and the choice of the discount factor. In particular, $\gamma=0.9$ does not yield a feasible certificate for any of the polynomial degrees considered, whereas choosing $\gamma$ sufficiently close to one substantially improves feasibility. This is consistent with the theoretical role of the discounted reach-avoid value function: a larger $\gamma$ provides a closer approximation to the undiscounted reach-avoid probability. Moreover, increasing the polynomial degree enlarges the search space and progressively improves the ability of the SDP \eqref{v_sos} to represent a suitable barrier certificate.

For $\gamma=0.999$, feasibility is first obtained at degree $10$ for
$\epsilon=0.70$, and higher probability thresholds become feasible as
the polynomial degree increases. In particular, at degree $20$, the
formulation is feasible for all four thresholds considered, up to
$\epsilon=0.83$. Increasing the discount factor to $\gamma=0.9999$
improves feasibility at several degrees; for example, $\epsilon=0.75$
is feasible already at degree $10$ and $\epsilon=0.80$ at degree $14$. These results illustrate that both the polynomial degree and the choice of discount factor affect the feasibility of the SDP
formulation, with larger degrees and discount factors closer to one
generally allowing less conservative probability thresholds.

\begin{table}
\caption{\centering Feasibility of SDP \eqref{v_sos} for Example~\ref{ex1}
(\ding{52}: feasible; \ding{55}: infeasible)}
\label{tab:sdp_ex1_set}
\centering
\begin{tabular}{|c|c|c|c|c|c|}
\hline
Degree&$\gamma$ & $\epsilon=0.70$ & $\epsilon=0.75$ & $\epsilon=0.80$  &$\epsilon=0.83$ 
\\
\hline
8 &0.9  & \ding{55} & \ding{55} & \ding{55} & \ding{55} \\
8 &0.999  & \ding{55} & \ding{55} & \ding{55} & \ding{55} \\
8 &0.9999  & \ding{55} & \ding{55} & \ding{55} & \ding{55} \\
10 &0.9  & \ding{55} & \ding{55} & \ding{55}& \ding{55} \\
10 &0.999  & \ding{52} & \ding{55} & \ding{55}& \ding{55} \\
10 &0.9999  & \ding{52} & \ding{52} & \ding{55}& \ding{55} \\
12 &0.9  & \ding{55} & \ding{55} & \ding{55} & \ding{55} \\
12 &0.999  & \ding{52} & \ding{55} & \ding{55} & \ding{55} \\
12 &0.9999  & \ding{52} & \ding{52} & \ding{55} & \ding{55} \\
14 &0.9  & \ding{55} & \ding{55} & \ding{55} & \ding{55}\\
14 &0.999  & \ding{52} & \ding{55} & \ding{55} & \ding{55}\\
14 &0.9999  & \ding{52} & \ding{52} & \ding{52} & \ding{55}\\
16 &0.9  & \ding{55} & \ding{55} & \ding{55} & \ding{55} \\
16 &0.999  & \ding{52} & \ding{55} & \ding{55} & \ding{55} \\
16 &0.9999  & \ding{52} & \ding{52} & \ding{52} & \ding{55} \\
18 &0.9 & \ding{55} & \ding{55} & \ding{55} & \ding{55} \\
18 &0.999 & \ding{52} & \ding{55} & \ding{55} & \ding{55} \\
18 &0.9999 & \ding{52} & \ding{52} & \ding{52} & \ding{55} \\
20 &0.9 & \ding{55} & \ding{55} & \ding{55} & \ding{55}   \\
20 &0.999 & \ding{52} & \ding{52} & \ding{52} & \ding{52}   \\
\hline
\textbf{MC} & \multicolumn{5}{c|}{\textbf{0.8445}} \\
\hline
\end{tabular}
\end{table}

We next consider the alternative barrier-like condition \eqref{w}, which eliminates the explicit discount factor and introduces an auxiliary function $w$. The corresponding SOS formulation is
 \begin{equation}
  \label{w_sos}
\begin{cases}
   v(\bm{x})-\epsilon+s_0(\bm{x})h_0(\bm{x})\in \sum[\bm{x}], \\
       \mathbb{E}_{\theta}[v(\bm{f}(\bm{x},\bm{\theta}))]-v(\bm{x})+s_1(\bm{x}) h(\bm{x})\\
       \qquad\qquad\qquad-s_2(\bm{x}) g(\bm{x}) \in \sum[\bm{x}],\\
        \mathbb{E}_{\theta}[w(\bm{f}(\bm{x},\bm{\theta}))]-w(\bm{x})-v(\bm{x})\\
        \qquad\quad+s_3(\bm{x}) h(\bm{x})-s_4(\bm{x}) g(\bm{x}) \in \sum[\bm{x}],\\
        1-v(\bm{x})+s_5(\bm{x})g(\bm{x})\in \sum[\bm{x}],\\
       -v(\bm{x})-s_6(\bm{x})h(\bm{x})+s_7(\bm{x})\widehat{h}(\bm{x})\in \sum[\bm{x}],\\
       s_0(\bm{x})\in \sum[\bm{x}],s_1(\bm{x})\in \sum[\bm{x}],\\
       s_2(\bm{x})\in \sum[\bm{x}],s_3(\bm{x})\in \sum[\bm{x}],\\
       s_4(\bm{x})\in \sum[\bm{x}],s_5(\bm{x})\in \sum[\bm{x}],\\
       s_6(\bm{x})\in \sum[\bm{x}],s_7(\bm{x})\in \sum[\bm{x}].
       \end{cases}
\end{equation}

The corresponding feasibility results are reported in Table \ref{tab:sdp_ex1_w}. In contrast to \eqref{v_sos}, this formulation does not require a prescribed discount factor. In this example, degree $10$ already yields feasible certificate for $\epsilon=0.7$ and $\epsilon=0.75$, while degree $12$ is sufficient for $\epsilon=0.80$ and degree $20$ for $\epsilon=0.83$. Thus, the $w$-based formulation avoids the sensitivity to the choice of $\gamma$ observed for \eqref{v_sos} and provides feasible certificates for several threshold at lower polynomial degrees.

\begin{table}
\caption{\centering Feasibility of SDP \eqref{w_sos} for Example~\ref{ex1}
(\ding{52}: feasible; \ding{55}: infeasible)}
\label{tab:sdp_ex1_w}
\centering
\begin{tabular}{|c|c|c|c|c|}
\hline
Degree & $\epsilon=0.70$ & $\epsilon=0.75$ & $\epsilon=0.80$  &$\epsilon=0.83$ 
\\
\hline
8  & \ding{55} & \ding{55} & \ding{55} & \ding{55} \\
10  & \ding{52} & \ding{52} & \ding{55}& \ding{55} \\
12  & \ding{52} & \ding{52} & \ding{52} & \ding{55} \\
14  & \ding{52} & \ding{52} & \ding{52} & \ding{55}\\
16  & \ding{52} & \ding{52} & \ding{52} & \ding{55} \\
18 & \ding{52} & \ding{52} & \ding{52} & \ding{55} \\
20 & \ding{52} & \ding{52} & \ding{52} & \ding{52}   \\
\hline
\textbf{MC} & \multicolumn{4}{c|}{\textbf{0.8445}} \\
\hline
\end{tabular}
\end{table}

\end{example}
 
\section{Conclusion}
In this paper, we studied uniform reach-avoid verification for stochastic discrete-time systems over compact initial sets. Under appropriate assumptions, we obtained the necessity result of barrier-like conditions involving a discount factor for uniform reach-avoid verification. We also established the necessity of an alternative barrier-like formulation involving an auxiliary function. Both formulations were demonstrated through SOS-based computations. 

Future work will focus on relaxing the assumptions and developing computational methods for higher-dimensional nonlinear stochastic systems.

\bibliographystyle{splncs04}
\bibliography{ref}

\end{document}